\documentclass[11pt,a4paper]{article}

\usepackage{algo}
\usepackage{graphicx}
\graphicspath{{figures/}}
\usepackage[english]{babel}
\usepackage{latexsym}
\usepackage{url}
\usepackage{amsfonts,amsmath,amsthm}
\usepackage{dsfont}
\usepackage{ifthen}
\usepackage{fancybox}
\usepackage{stmaryrd}
\usepackage{xspace}

\usepackage{color}
\definecolor{blu3}{rgb}{.1,.0,.4}
\usepackage{hyperref}
\hypersetup{colorlinks=true, linkcolor=blu3, urlcolor=blu3, citecolor=blu3, pdfpagemode=UseNone, pdfstartview=} 

\if10     
\usepackage[mathlines]{lineno}
\newcommand*\patchAmsMathEnvironmentForLineno[1]{%
  \expandafter\let\csname old#1\expandafter\endcsname\csname #1\endcsname
  \expandafter\let\csname oldend#1\expandafter\endcsname\csname end#1\endcsname
  \renewenvironment{#1}%
     {\linenomath\csname old#1\endcsname}%
     {\csname oldend#1\endcsname\endlinenomath}}%
\newcommand*\patchBothAmsMathEnvironmentsForLineno[1]{%
  \patchAmsMathEnvironmentForLineno{#1}%
  \patchAmsMathEnvironmentForLineno{#1*}}%
\AtBeginDocument{%
\patchBothAmsMathEnvironmentsForLineno{equation}%
\patchBothAmsMathEnvironmentsForLineno{align}%
\patchBothAmsMathEnvironmentsForLineno{flalign}%
\patchBothAmsMathEnvironmentsForLineno{alignat}%
\patchBothAmsMathEnvironmentsForLineno{gather}%
\patchBothAmsMathEnvironmentsForLineno{multline}%
}
\linenumbers
\fi

\usepackage[margin=1.1in]{geometry}

\newtheorem{theorem}{Theorem}

\newtheorem{lemma}[theorem]{Lemma}

\def\DEF#1{\textbf{\emph{#1}}}
\def\eps{\varepsilon}

\DeclareMathOperator{\loglog}{loglog}
\DeclareMathOperator{\sky}{sky}
\DeclareMathOperator{\opt}{opt}
\DeclareMathOperator{\next}{succ}
\DeclareMathOperator{\prev}{pred}
\DeclareMathOperator{\nrp}{nrp}
\def\output{\emph{output}}

\begin{document}

\title{Faster Distance-Based Representative Skyline and\\
	$k$-Center Along Pareto Front in the Plane}

\author{Sergio Cabello\thanks{Faculty of Mathematics and Physics, University of Ljubljana, Slovenia, 
	and Institute of Mathematics, Physics and Mechanics, Slovenia. 
	Supported by the Slovenian Research Agency (P1-0297, J1-9109, J1-1693, J1-2452).
	Email address: \texttt{sergio.cabello@fmf.uni-lj.si}.}}

\maketitle

\begin{abstract}
  We consider the problem of computing the \emph{distance-based representative skyline} in the plane,
  a problem introduced by Tao, Ding, Lin and Pei [Proc. 25th IEEE International Conference on Data 
  Engineering (ICDE), 2009] and independently considered by 
  Dupin, Nielsen and Talbi [Optimization and Learning - Third International Conference, OLA 2020] 
  in the context of multi-objective optimization.
  Given a set $P$ of $n$ points in the plane and a parameter $k$, the task is to select $k$
  points of the skyline defined by $P$ (also known as Pareto front for $P$)
  to minimize the maximum distance from the points of the skyline to the selected points.  
  We show that the problem can be solved in $O(n\log h)$ time, where $h$ is the 
  number of points in the skyline of $P$. 
  We also show that the decision problem can be solved in $O(n\log k)$ time and
  the optimization problem can be solved in $O(n \log k  + n \loglog n)$ time.
  This improves previous algorithms and is optimal for a large range of values of $k$. 
 
\end{abstract}

\section{Introduction}
For each point $p$ in the plane, let $x(p)$ and $y(p)$ denote its $x$- and $y$-coordinate, respectively.
Thus $p=(x(p),y(p))$.
A point $p$ \DEF{dominates} a point $q$ if $x(p)\ge x(q)$ and $y(p)\ge y(q)$.
Note that a point dominates itself.
For a set of points $P$, its \DEF{skyline} is the subset of points of $P$ that are not
dominated by any other point of $P$.
We denote by $\sky(P)$ the set of skyline points of $P$. Formally
\[
	\sky(P) = \{ p\in P\mid \forall q\in P\setminus\{ p\}:~ x(q)<x(p) \text{ or } y(q)<y(p)\}.
\]
See the left of Figure~\ref{fig:skyline} for an example.
The skyline of $P$ is also called the \DEF{Pareto front} of $P$.
A set $P$ of points is a skyline or a Pareto front if $P=\sky(P)$.
Note that some authors exchange the direction of the inequalities, depending
on the meaning of the coordinates in the application domain.

\begin{figure}
\centering
	\includegraphics[page=1,scale=1.1]{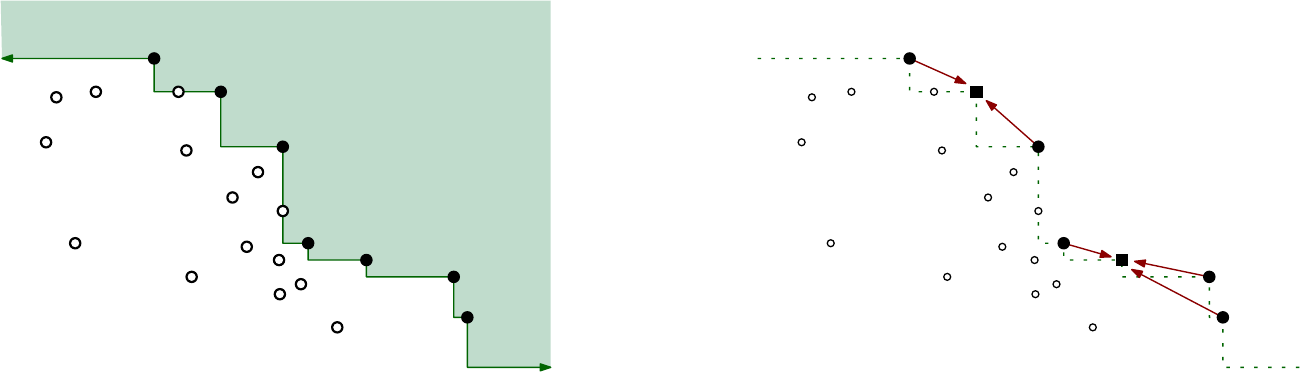}
	\caption{Left: Example of point set $P$ with
				the points of $\sky(P)$ marked as filled dots;
				the shaded region has to be empty of points of $P$.
			Right: if $Q$ is the two-point set marked with squares,
				then the length of the longest arrow is $\psi(Q,P)$.}
	\label{fig:skyline}
\end{figure}

For each subset of points $Q\subseteq \sky(P)$ we define
\[
	\psi(Q,P) := \max_{p\in \sky(P)}~~ \min_{q\in Q} ~~d(p,q),
\]
where $d(p,q)$ denotes the \DEF{Euclidean distance} between $p$ and $q$.
Note that $\psi(Q,P) = \psi(Q,\sky(P))$ and $\psi(\sky(P),\sky(P))=0$.
An alternative point of view for $\psi(Q,P)$ is the following:
it is the smallest value $\lambda$ such
that disks centered at the points of $Q$ with radius $\lambda$
cover the whole $\sky(P)$. 
See Figure~\ref{fig:skyline}, right.
One can consider $\psi(Q,P)$ as how big error we make when approximating $\sky(P)$ by $Q$.

In this paper we provide efficient algorithms for the following optimization problem: 
for a given point set $P$ in the plane and a positive integer $k$, compute
\[
	\opt(P,k) :=  \min \{ \psi(Q,P)\mid Q\subseteq \sky(P) \text{ and } |Q|\le k \}.
\]
A \DEF{feasible solution} is any subset $Q\subseteq \sky(P)$ with $|Q|\le k$.
An \DEF{optimal solution} for $\opt(P,k)$ is a feasible solution $Q^*$ 
such that $\psi(Q^*,P)=\opt(P,k)$. Thus, we use $\opt(P,k)$ both as the optimal
value and the description of the problem.
Note that $\opt(P,k)=\opt(\sky(P),k)$.

We will also consider the \DEF{decision problem}:
given a set $P$ of points in the plane, a positive integer $k$ and a real value $\lambda\ge 0$,
decide whether $\opt(P,k)\le \lambda$ or $\opt(P,k)> \lambda$.

\paragraph{Motivation.}
Skylines (or Pareto fronts) can be considered for arbitrary dimensions and play a basic
role in several areas.
B{\"{o}}rzs{\"{o}}nyi, Kossmann and Stocker~\cite{BorzsonyiKS01} advocated their
use in the context of databases, where they have become an ubiquitous tool; 
see for example~\cite{Peng2016,StefanidisKP11}.
They also play a key role in multiobjective optimization and decision analysis;
see for example the surveys~\cite{EmmerichD18,LiY19,0001MBC14}.

Quite often the skyline contains too many points. The question is then how to select
a representative subset of the points in the skyline.
The representative may be a final answer or just an intermediary step,
as it happens for example in evolutionary algorithms, where diversity plays a key role.

As one can imagine, several different meaningful definitions of \emph{best representative}
subset are reasonable, and different applications advocate for different measures
on how good the representatives are. The use of randomization is also relevant in the context
of evolutionary algorithms.

The problem we are considering, $\opt(P,k)$, 
was introduced in the context of databases by Tao et al.~\cite{TaoDLP09}
under the name of \emph{distance-based representative} of $\sky(P)$.
Although we have presented the problem in the plane, it can
be naturally extended to arbitrary dimensions.
Tao et al. argue the benefit of using such representatives.
One of the key properties is that it is not sensible to the density of $P$,
a desirable property when $P$ is not uniformly distributed.

The very same problem, $\opt(P,k)$,
was considered in the context of optimization
by Dupin, Nielsen and Talbi~\cite{DupinNT20}, who noted the connection to 
clustering, multiobjective optimization and decision analysis. 
They noticed that $\opt(P,k)$ is the (discrete) $k$-center problem
for $\sky(P)$.
Clustering problems in general, and $k$-center problem in particular~\cite{centers,HochbaumS85} 
are fundamental in several areas;
thus it is relevant to find particular cases that can be solved efficiently.

Another very popular and influential measure to select representatives
in the context of databases was introduced by Lin et al.~\cite{LinYZZ07}. 
In this case, we want to select $k$ points of $\sky(P)$
so as to maximize the number of points from $P$ dominated by some of the selected points.
See our work~\cite{ChoiCA19} for a recent algorithmic improvement in the planar case.
In the context of evolutionary algorithms one of the strongest measures
is the hypervolume of the space dominated by the selected points~\cite{BeumeNE07}
with respect to a reference point, that can be taken to be the origin.
As it can be seen in the survey by Li and Yao~\cite{LiY19}, an overwhelming amount
of different criteria to select representatives from $\sky(P)$ has been considered
in the area of multiobjective optimization.

\paragraph{Previous algorithms.}
We will be using three parameters to analyze the time complexity of algorithms.
We will use $n$ for the size of $P$, $k$ for the bound on the number of points to be selected,
and $h$ for the number of points in $\sky(P)$.
Unless explicitly mentioned, the discussion is for the plane.

We start reviewing the computation of the skyline.
Kung, Luccio and Preparata~\cite{KungLP75} showed that the skyline of $P$ in the plane
can be computed in $O(n \log n)$ time, and this is asymptotically optimal.
Kirkpatrick and Seidel~\cite{KirkpatrickS85} noted that the lower bound of $\Omega(n \log n)$
does not hold when $h$ is small, and
they provided an algorithm to compute $\sky(P)$ (in the plane) in $O(n \log h)$ time.
Nielsen~\cite{Nielsen96} provided an algorithm to compute the top $k$ skylines,
that is, $\sky(P)$, $\sky(P\setminus \sky(P))$, $\dots$ up to $k$ iterations.
The running time is $O(n \log H)$, where $H$ is the size of the top $k$ skylines together.

Kirkpatrick and Seidel~\cite{KirkpatrickS85} also showed that computing the skyline of $P$ 
takes $\Omega(n \log h)$ time in the algebraic decision tree model. 
In fact, they showed that, for a given $h'$, deciding whether $P$ has exactly $h'$
points in the skyline takes $\Omega( n\log h')$ time.
This implies that finding any $k$ points in the skyline takes $\Omega (n \log k)$ time.
Indeed, if we could find $k$ points in the skyline in $o(n\log k)$, then
we could run the test for $k=h'$ and $k=h'+1$ and decide in $o(n \log(k+1))=o(n\log h')$
time whether $\sky(P)$ has exactly $h'$ points, contradicting the lower bound.

Now we move onto computing $\opt(P,k)$.
Tao et al.~\cite{TaoDLP09} showed that the problem can be solved in $O(kh^2)$ assuming
that the skyline is available and sorted by $x$-coordinate. Thus, it takes $O(n\log h + kh^2)$ time.
In the full version of the work~\cite{TaoLDLP_long} they improved the time
bound of the second phase to $O(kh)$, which implies a time bound of $O(n\log h + kh)$.
They showed that the 3d version of the problem is NP-hard.

Dupin, Nielsen and Talbi~\cite{DupinNT20} solve the $\opt(P,k)$ problem in 
$O(kh\log^2h)$ time, again assuming that the skyline is available and sorted.
Thus, starting from $P$, the time bound is $O(n\log h + kh\log^2h)$.
They also provide a linear-time algorithm for $\opt(P,1)$ and an
algorithm with running time $O(h\log h)$ for $\opt(P,2)$, again assuming that the 
skyline is available and sorted by $x$-coordinate.

Mao et al.~\cite{MaoCLYL17} provided an exact algorithm for 
$\opt(P,k)$ that takes $O(k^2 \log^3 h)$ time, assuming that the skyline is stored
for binary searches, and an approximation with error of $\eps$ in $O(k \log^2 h \log(T/\eps))$,
where $T$ is the distance between the extreme points of the skyline.
Again, if we start from $P$, we should add the time $O(n \log h)$ to both algorithms to
compute the skyline.
The decision problem is solved in linear time in the work by
Yin and Wei and Liu~\cite{YinWL19}, assuming that the skyline is available.

If the points on the skyline are collinear and sorted,
we can compute $\opt(P,k)$ in linear time
using the algorithm of Frederickson~\cite{Frederickson91,Frederickson91b}.
See~\cite{abs-1711-00599} for a recent account that may be easier to understand.
The author believes that the same approach can be used to
solve $\opt(P,k)$ in linear time, once the skyline is available sorted.
However, as the details in those algorithms are very complicated,
the author is reluctant to make a firm claim.

\paragraph{Our contribution.}
We show that $\opt(P,k)$ can be solved in $O(n \log h)$ time.
We provide a quite simple algorithm for this. 
Compared to a possible adaptation of~\cite{Frederickson91,Frederickson91b},
we get the added value of simplicity within the same time
bound that we need to compute the skyline itself.
If the skyline is already available, the running time becomes $O(h \log h)$.
In all cases we improve all previous works considering algorithms for the problem $\opt(P,k)$
explicitly~\cite{DupinNT20,MaoCLYL17,TaoDLP09,TaoLDLP_long}.

At first, this may seem optimal because computing the skyline of $P$ takes
$\Omega( n\log h)$ time.
However, do we really need to compute the skyline? 
After all, we only need to select a particular subset of $k$ points from the skyline, 
and this has a lower bound of $\Omega(n\log k)$ time, as mentioned above.
We show that the decision problem for $\opt(P,k)$ 
can be solved in $O(n\log k)$ time. 
This is asymptotically optimal if we want to find a solution because 
such a solution has to find $k$ points in $\sky(P)$.
In fact, we show that with some additional preprocessing we
can solve several decision problems. For example, we can solve $O(k^2)$
decision problems in $O(n\log k)$ time altogether.
We combine the decision problem with parametric search and selection
in sorted arrays to show that 
$\opt(P,k)$ can be solved in $O(n\log k + n\loglog n)$ time.
This is asymptotically optimal whenever $\log k = \Omega(\loglog n)$,
if we want to find also an optimal solution.

Finally, we provide additional results for the case when $k$ is very small.
The driving question here is, for example, how fast can we solve $\opt(P,15)$.
We show that a $(1+\eps)$-approximation can be computed in 
$O(kn + n\loglog (1/\eps))$ time.
We obtain this by finding a $2$-approximation in $O(kn)$ time and then making
use of repetitive binary searches using the decision problem.

Besides improving previous algorithms,
we believe that the conceptual shift of solving $\opt(P,k)$ without 
computing $\sky(P)$ explicitly is a main contribution of our work.
Nevertheless, even when $\sky(P)$ is available, we improve
previous works, with the possible exception of~\cite{Frederickson91}.
We provide very detailed description of our algorithms
to indicate that they are implementable.

We start the paper discussing in detail an optimal algorithm for computing
the skyline in $O(n\log h)$ time.
The algorithm uses the techniques from Chan~\cite{Chan96} and Nielsen~\cite{Nielsen96}.
The purpose is two-fold. First, we believe that the presentation of this
part is pedagogical, slightly simpler than the presentation by Nielsen
because we try to solve a simpler problem, and may find interest within
the community interested in computing with skylines in the plane.
Second, it provides the basic idea for our later approach, and thus it should
be discussed to some level of detail in any case.
We decided to provide a self-contained presentation of this part.

\paragraph{Organization of the paper.}
In Section~\ref{sec:basictools} we provide some basic observations.
In Section~\ref{sec:skyline} we explain a simple, optimal algorithm
to compute the skyline.
Section~\ref{sec:optimization_skyline} is devoted to the problem
of finding the skyline representatives through the computation of
the skyline, while Section~\ref{sec:optimization_no_skyline}
looks at the problem without computing the skyline.
Results for very small $k$ are considered in Section~\ref{sec:smallk}.
We conclude in Section~\ref{sec:conclusions},
where we also provide further research questions.

\section{Basic tools}
\label{sec:basictools}
We are usually interested in the order of the points along the skyline $\sky(P)$. 
For this we store the points of $\sky(P)$ sorted by increasing $x$-coordinate 
in an array. We could also use any other data structure that allows us to make 
binary searches among the elements of $\sky(P)$, such as a balanced binary search tree. 
Note that sorting $\sky(P)$ by $x$- or by (decreasing) $y$-coordinate is equivalent.

Quite often we use an expression like ``select the 
highest point among $q_1,\dots,q_t$, breaking ties in favor of larger $x(\cdot)$''.
This means that we take the point among $q_1,\dots,q_t$ with largest $y$-coordinate;
if more points have the largest $y$-coordinate $y_{\max}$, then
we select among those with $y(q_i)=y_{\max}$ 
the one with largest $x$-coordinate.
This can be done iterating over $i=1,\dots, t$, storing the best point $q_*$ so far,
and updating  
\[
	q_*\qlet q_i \text{ if and only if }
	y(q_*)< y(q_i) \text{ or }~ (y(q_*)=y(q_i) \text{ and } x(q_*)<x(q_i) ).
\]
This wording appears to take care of the cases where more points have
the same $x$- or $y$-coordinate. 
Let us explain an intuition that may help understanding how ties
are broken.
Conceptually, think that each point $p=(x,y)\in P$ is replaced
by the point $p_\eps=(x+y\eps, y+x\eps)$ for an infinitesimally small $\eps>0$,
and let $P_\eps$ be the resulting point set.
A point $p\in P$ belongs to $\sky(P)$ if and only if $p_\eps$ belongs to $\sky(P_\eps)$. 
Moreover, no two coordinates are the same in $P_\eps$, if $\eps>0$ is sufficiently small.
Thus, whenever we have to break ties, we think what we would do for $P_\eps$.

For correctness and making binary searches we will often use the following
monotonicity property.
\begin{lemma}
\label{le:monotone}
	If $p,q,r$ are points of $\sky(P)$ and $x(p)<x(q)<x(r)$,
	then $d(p,q)<d(p,r)$.
\end{lemma}
\begin{proof}
	From the hypothesis we have $y(p)>y(q)>y(r)$ and therefore
	\begin{align*}
		\bigl( d(p,q) \bigr) ^2 &= \bigl( x(q)-x(p) \bigr)^2 + \bigl( y(q)-y(p) \bigr)^2\\
		& < \bigl( x(r)-x(p) \bigr)^2 + \bigl( y(r)-y(p) \bigr)^2
		= \bigl( d(p,r) \bigr) ^2.
	\qedhere
	\end{align*}
\end{proof}

A first consequence of this lemma 
is that any disk centered at a point of $\sky(P)$
contains a contiguous subsequence of $\sky(P)$.
In particular, if $\sky(P)$ is 
stored for binary searches along $x$-coordinates,
then we can perform binary searches  
to find the boundary elements of $\sky(P)$ contained in 
any given disk \emph{centered at a point of} $\sky(P)$. 
(The claim is not true for disks centered at arbitrary points of $P$.)

\medskip

For an $x$-coordinate $x_0$ and a set of points $Q$, let $\next(Q,x_0)$
be the leftmost point of $Q$ in the open halfplane to the right of the vertical line $x=x_0$.
Similarly, we denote by $\prev(Q,x_0)$
the rightmost point of $Q$ in the open halfplane to the left of the vertical line $x=x_0$.
Thus
\begin{align*}
	\next(Q,x_0) &= \arg\min \{ x(q)\mid q\in Q,~ x(q)>x_0 \},\\
	\prev(Q,x_0) &= \arg\max \{ x(q)\mid q\in Q,~ x(q)<x_0 \}.
\end{align*}
We will use $\next(Q,x_0)$ and $\prev(Q,x_0)$ when $Q$ is a skyline;
see Figure~\ref{fig:next}.
In our use we will also take care that $\next(Q,x_0)$ and $\prev(Q,x_0)$
are well-defined, 
meaning that there is some point of $Q$ on each side of the vertical line $x=x_0$.

When $Q$ is a skyline stored for binary searches along $x$-coordinates,
the points $\next(Q,x_0)$ and $\prev(Q,x_0)$ can be obtained in $O(\log |Q|)$ time.

Looking into the case of skylines we can observe that
\begin{align*}
	\next(\sky(P),x_0)&= \arg\max \{ y(p)\mid p\in P,~ x(p)>x_0 \}, \\
	&~~~~~~~~~~~~~~~~~~~~\qquad\text{ breaking ties in favor of larger $x(\cdot)$},
\end{align*}
while using $y_0= \max \{ y(p)\mid p\in P,~ x(p)\ge x_0\}$
we have
\begin{align*}
	\prev(\sky(P),x_0) &= \arg\max \{ x(p)\mid p\in P,~ y(p)> y_0 \},\\
	&~~~~~~~~~~~~~~~~~~~~\qquad\text{ breaking ties in favor of larger $y(\cdot)$}. 	
\end{align*}

\begin{figure}
\centering
	\includegraphics[page=2,scale=1.1]{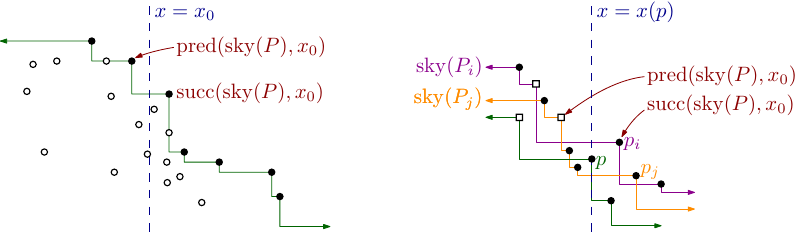}
	\caption{Left: $\next(\sky(P),x_0)$ and $\prev(\sky(P),x_0)$ for a set of points $P$.
			Right: Computing $\next(\sky(P),x_0)$ and $\prev(\sky(P),x_0)$ when
				we have $\sky(P_1),\dots,\sky(P_t)$ for $P=P_1\cup\dots\cup P_t$.
				The points $q_1,\dots,q_t$ are marked with empty squares.}
	\label{fig:next}
\end{figure}

\medskip

In several of our algorithms we will use a divide and conquer approach.
The following observation is one of the simple but key observations.
See Figure~\ref{fig:next}, right.

\begin{lemma}
\label{le:next}
	Assume that $P$ is the union of $P_1,\dots,P_t$.
	For each real value $x_0$, the point $\next(\sky(P),x_0)$ is 
	the highest point among $\next(\sky(P_1),x_0),\dots, \next(\sky(P_t),x_0)$,
	breaking ties in favor of larger $x(\cdot)$.
	Moreover, if for each $i\in \{1,\dots,t\}$ the points of $\sky(P_i)$ 
	are sorted by $x$-coordinate for binary searches, 
	then we can compute $\next(\sky(P),x_0)$ in $O(t+\sum_i \log |P_i|)$ time.
\end{lemma}
\begin{proof}
	The point $\next(\sky(P),x_0)$ is the highest point of $P$ in the halfplane $x>x_0$,
	breaking ties in favor of larger $x(\cdot)$.
	Since
	\begin{align*}
		y\bigl(\next(\sky(P),x_0)\bigr) &= \max \{ y(p)\mid p\in P,~ x(p)>x_0 \} \\
			& = \max_{i=1,\dots, t}\; \max \{ y(p_i)\mid p_i\in P_i,~ x(p_i)>x_0 \}\\
			& = \max_{i=1,\dots, t}\; y\bigl( \next(\sky(P_i),x_0) \bigr)\},
	\end{align*}
	the claim about the characterization of $\next(\sky(P),x_0)$ follows.
	The algorithmic claim is obtained by making a binary search in each $\sky(P_i)$
	to obtain $\next(\sky(P_i),x_0)$ and choosing the highest point; in case of ties
	we select the one with largest $x$-coordinate.
\end{proof}

We do not provide pseudocode for the simple algorithm of Lemma~\ref{le:next},
but the idea will be used inside pseudocode provided later.
We also have the following tool, which is slightly more complicated.
It tests whether a point belongs to the skyline and computes the
predecessor. 
Pseudocode under the name of {\sc TestSkylineAndPredecessor}
is described in Figure~\ref{fig:TestSkylineAndPredecessor}.

\begin{lemma}
\label{le:decisionmany}
	Assume that $P$ is the union of $P_1,\dots,P_t$
	and for each $i\in \{ 1,\dots, t\}$ the points of $\sky(P_i)$ are 
	sorted by $x$-coordinate for binary searches.
	For any point $p\in P$, the algorithm {\sc TestSkylineAndPredecessor}
	solves the following two problems in $O(t+ \sum_i \log |P_i|)$ time: 
	\begin{itemize}
		\item decide whether $p\in \sky(P)$;
		\item compute $\prev(\sky(P),x(p))$.
	\end{itemize}
\end{lemma}	
\begin{proof}
	For each $i\in \{ 1,\dots, t\}$ we use a binary search along $\sky(P_i)$ to find
	the point $p_i$ of $\sky(P_i)$ with smallest $x$-coordinate among those with $x(p)\ge x(p_i)$.
	Let $p_0$ be the highest point among $p_1,\dots, p_t$, breaking ties in favor
	of those with larger $x$-coordinate.
	Note that $p_0$ is the highest point of $\sky(P)$ in $x\ge x(p)$;
	the proof of Lemma~\ref{le:next} can be trivially adapted for this.
	The point $p$ is in $\sky(P)$ if and only if $p=p_0$.

	Computing $\prev(P,x(p))$ is similar to computing $\next(P,x(p))$, if we exchange
	the roles of $x$ and $y$-coordinates. 
	See Figure~\ref{fig:next}, right, for an example.
	For each $i\in \{ 1,\dots, t\}$, we use a binary search along $\sky(P_i)$ 
	\emph{with respect to the $y$-coordinate} and key $y(p_0)$, 
	to find the point $q_i\in \sky(P_i)$ with 
	smallest $y$-coordinate among those with $y(q_i) > y(q_0)$. 
	The point among $q_1,\dots, q_t$ with largest $x$-coordinate (and breaking
	ties in favor of the largest $y$-coordinate) is $\prev(\sky(P),x(p))$.
	(As seen in the example of Figure~\ref{fig:next}, right,
	$\prev(\sky(P),x(p))$ is not necessarily any of the
	points among $\prev(\sky(P_1),x(p)),\dots, \prev(\sky(P_t),x(p))$.)
	
	Algorithm {\sc TestSkylineAndPredecessor} implements the idea of this proof.
\end{proof}

\begin{figure}
\centering
\fbox{~~~\parbox{0.89\linewidth}{
	\begin{algorithm}{{\sc TestSkylineAndPredecessor}}{
	\qinput{$\sky(P_1),\dots,\sky(P_t)$ stored for binary searches and a point $p$.} 
	\qoutput{A pair $(X,q)$ where $X$ is True if $p\in \sky(P_1\cup\dots \cup P_t)$ and False otherwise,
		and $q=\prev( \sky(P_1\cup\dots \cup P_t), x(p))$.}}
	\qfor $i=1,\dots, t$ \qdo\\
		$p_i\qlet$ point of $\sky(P_i)$ in $x\ge x(p)$ with smallest $x$-coordinate, using binary search
	\qrof\\
	$p_0\qlet$ highest point among $p_1,\dots,p_t$, breaking ties in favor of larger $x(\cdot)$\\
	\qfor $i=1,\dots, t$ \qdo\\
		$q_i\qlet$ point of $\sky(P_i)$ in $y> y(p_0)$ with smallest $y$-coordinate, using binary search
	\qrof\\
	$q_0\qlet$ rightmost point among $q_1,\dots,q_t$, breaking ties in favor of larger $y(\cdot)$\\
	\qreturn ($p=p_0?$, $q_0$)
	\end{algorithm}}}
	\caption{Algorithm described in Lemma~\ref{le:decisionmany}
		for testing whether $p\in \sky(P)$ and computing $\prev(\sky(P),x(p))$.}
\label{fig:TestSkylineAndPredecessor}
\end{figure}

\section{Computing the skyline optimally}
\label{sec:skyline}
In this section we show that the skyline of a set $P$ of $n$ points can
be computed in $O(n\log h)$, where $h=|\sky(P)|$. 
The algorithm is simple and reuses the ideas exploited
by Chan~\cite{Chan96} and Nielsen~\cite{Nielsen96}.
We include it because of its simplicity and because we will modify it later on.
Compared to Nielsen~\cite{Nielsen96}, our algorithm is slightly simpler because
we are computing a slightly simpler object.

It is well-known and very easy to see that $\sky(P)$ can be computed in $O(n\log n)$ time.
Indeed, we just need to sort the points lexicographically: a point $p$ precedes
a point $q$ if $x(p)< x(q)$ or if $x(p)=x(q)$ and $y(p)<y(q)$. 
Then we make a pass over the reversed list of sorted points, 
and maintain the point with largest $y(\cdot)$ scanned so far. The skyline of $P$ is the list
of local maxima we had during the procedure.
See Figure~\ref{fig:simple_algorithm} for the idea.
Details are given as Algorithm {\sc SlowComputeSkyline} in Figure~\ref{alg:SlowComputeSkyline}.
Note that the lexicographic order is important for correctness: 
if two points $p_i$ and $p_{i-1}$ have the same $x$-coordinate,
then the point $p_i$ has larger $y$-coordinate and thus $p_{i-1}$
does not make it to the skyline.
We will be using that $\sky(P)$ is returned sorted by $x$-coordinate.

\begin{figure}
\centering
	\includegraphics[page=3,scale=1.1]{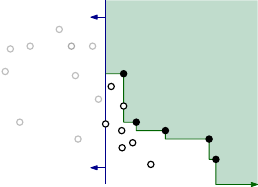}
	\caption{Idea in {\sc SlowComputeSkyline}.}
	\label{fig:simple_algorithm}
\end{figure}

\begin{figure}
\centering
\fbox{~~~\parbox{0.8\linewidth}{
	\begin{algorithm}{{\sc SlowComputeSkyline}}{
	\qinput{A set $P$ of $n$ points.} 
	\qoutput{$\sky(P)$ sorted by $x$-coordinate.}}
		$p_1,\dots,p_n \qlet$ sort $P$ lexicographically increasing\\
		$\output\qlet$ new list with $p_n$\\
		$p\qlet p_n$ \qquad (* last point added to $\output$ *)\\
		\qfor $i=n-1,\dots,1$ \qdo\\
			\qif $y(p_i)> y(p)$ \qthen\\
				$p\qlet p_i$\\
				append $p_i$ to $\output$
			\qfi
		\qrof\\
		\qreturn $\output$ \\(* we may reverse $\output$ if wished by increasing $x$-coordinate *)
	\end{algorithm}}}
	\caption{Algorithm to compute $\sky(P)$ in $O(n\log n)$ time.}
	\label{alg:SlowComputeSkyline}
\end{figure}

We now describe the algorithm computing $\sky(P)$ in $O(n \log h)$ time. 
The core of the procedure is a subroutine, called {\sc ComputeSkylineBounded}
and described in Figure~\ref{alg:ComputeSkylineBounded},
that takes as input $P$ and an integer value parameter $s\ge 1$.
The subroutine returns $\sky(P)$ when $h\le s$, and a warning ``\emph{incomplete}'',
when $h>s$. The role of $s$ is that of a rough guess for the size of $\sky(P)$.

The idea of the subroutine {\sc ComputeSkylineBounded} is as follows.
First we compute a value $M$ that is an upper bound on the absolute values of all coordinates.
We will use the points $(-M,M)$ and $(M,-M)$ as dummy points to avoid having boundary cases.
Let $\tilde P=P\cup \{ (-M,M), (M,-M) \}$.
We split $P$ into $t\approx n/s$ groups of points, $P_1,\dots, P_t$,
each with roughly $s$ points, and add the dummy points $(-M,M)$ and $(M,-M)$ to each group $P_i$.
We then compute the skyline of each group $P_i$
using an algorithm with time complexity $O(n\log n)$, for example {\sc SlowComputeSkyline}.
Each outcome $\sky(P_i)$ is stored in an array sorted by increasing $x$-coordinate,
so that we can perform binary searches along $\sky(P_i)$.

Now we can compute $\sky(\tilde P)=\sky(P)\cup \{ (-M,M), (M,-M) \}$.
Assuming that we are at a point $p$ of $\sky(\tilde P)$, 
we can compute the next point $p'=\next(\sky(\tilde P),x(p))$ along $\sky(\tilde P)$
by taking the highest among $\next(\sky(P_1),x(p)),\dots, \next(\sky(P_t),x(p))$;
ties are broken by taking the point with largest $x$-coordinate. 
The correctness of this is formalized in Lemma~\ref{le:next}.
The procedure is started from the dummy point $(-M,M)$, 
and iteratively compute the next point $s$ times. 
If at some point we reach the dummy point $(M,-M)$, we know that we computed
the whole $\sky(P)$ (we do not report $(-M,M)$ and $(M,-M)$ as part of the output).
If after $s+1$ iterations we did not arrive to the dummy point $(M,-M)$,
then we have computed at least $s+1$ points of $\sky(P)$
and correctly conclude that $|\sky(P)|> s$.

\begin{figure}
\centering%
\fbox{~~~\parbox{0.94\linewidth}{
	\begin{algorithm}{{\sc ComputeSkylineBounded}}{
	\qinput{A set $P$ of points, a positive integer $s$.} 
	\qoutput{skyline for $P$, if it contains at most $s$ points, ``\emph{incomplete}'' otherwise.}}
		$M \qlet 1+ \max \bigcup_{p\in P}\{|x(p)|,|y(p)| \}$\\
		$t \qlet \lceil n/s \rceil$\\
		split $P$ into $t$ groups $P_1,\dots,P_t$, each of at most $s$ points\\
		\qfor $i=1,\dots,t$ \qdo\\
			append the points $(-M,M)$ and $(M,-M)$ to $P_i$ \qquad (* dummy end points *)\\
			$S_i \qlet$ {\sc SlowComputeSkyline}($P_i$)\\
			store $S_i$ for binary searches along $x$-coordinate
		\qrof\\
		$\output\qlet$ new empty list\\
		$p\qlet (-M,M)$ \qquad (* dummy starting point *)\\
		\qrepeat $s+1$ times\\
			\qfor $i=1,\dots,t$ \qdo\\
				$p_i\qlet \next(S_i,x(p))$, using binary search
			\qrof\\
			$p\qlet$ highest point among $p_1,\dots, p_t$, breaking ties in favor of larger $x(\cdot)$\\
			\qif $x(p)=M$ \qthen (* we arrived to $p=(M,-M)$ *)\\
				\qreturn $\output$\\
			\qelse\\
				append $p$ to $\output$
			\qfi
		\qtaeper\\
		\qreturn ``\emph{incomplete}'' (* because skyline has more than $s$ points *)
	\end{algorithm}}}%
	\caption{Algorithm to compute $\sky(P)$ in $O(n\log s)$ time, if $h\le s$.}
	\label{alg:ComputeSkylineBounded}
\end{figure}

\begin{lemma}
	Algorithm {\sc ComputeSkylineBounded} has the following properties:
	\begin{itemize}
		\item {\sc ComputeSkylineBounded}$(P,s)$ returns $\sky(P)$, if $|\sky(P)|\le s$,
		\item {\sc ComputeSkylineBounded}$(P,s)$ returns ``\emph{incomplete}'', if $|\sky(P)|> s$,
		\item {\sc ComputeSkylineBounded}$(P,s)$ takes $O(n\log s)$ time.
	\end{itemize}
\end{lemma}
\begin{proof}
	Because of Lemma~\ref{le:next}, in lines 11-13 we correctly replace $p$
	by $\next(\sky(\tilde P),x(p))$ in each iteration of the repeat loop.
	If we reach $p=(M,-M)$ at some point, then we know that we computed the whole $\sky(P)$.
	If after $s+1$ iterations we did not reach $p=(M,-M)$, then we computed $s+1$ points
	of $\sky(P)$ and the algorithm returns ``\emph{incomplete}''.
	
	It remains to bound the running time.
	For each $P_i$ we spend $O(s \log s)$ time to compute $\sky(P_i)$.
	This means that we spend $O((n/s)\cdot s \log s)= O(n \log s)$ time to compute 
	all $\sky(P_1), \dots,\sky(P_t)$.
	In each iteration of the repeat loop we make $O(t)=O(n/s)$ binary searches among 
	$O(s)$ elements, because each $S_i$ has $O(|P_i|)=O(s)$ elements.
	Thus, in one iteration we compute the points $p_1,\dots,p_t$ in $O(t\log s)=O((n/s)\log s)$ time,
	and selecting $p$ then takes additional $O(n/s)$ time.
	We conclude that each iteration takes $O((n/s)\log s)$ time, and since there are $s+1$ iterations,
	the time bound follows.
\end{proof}

Note that Algorithm {\sc ComputeSkylineBounded} can also be used to decide
whether $h\le s$ for a given $s$ in $O(n\log s)$ time.

To compute the skyline we make repetitive use of {\sc ComputeSkylineBounded}$(P,s)$ 
increasing the value of $s$ doubly exponentially: if $s$ is too small, 
we reset $s$ to the value $s^2$, and try again. The intuition is that
we want to make an exponential search for the value $\log h$ using terms of the form $\log s$. 
Thus, when we are not successful in computing the whole skyline because $s<h$, 
we want to set $s$ to double the value of $\log s$.
See the algorithm {\sc OptimalComputeSkyline} in Figure~\ref{alg:OptimalComputeSkyline}.

\begin{figure}
\centering
\fbox{~~~\parbox{0.6\linewidth}{
	\begin{algorithm}{{\sc OptimalComputeSkyline}}{
	\qinput{A set $P$ of points.} 
	\qoutput{skyline for $P$.}}
		$\output\qlet $``\emph{incomplete}''\\
		$s\qlet 4$\\
		\qwhile $\output = $``\emph{incomplete}''\qdo\\
			$\output \qlet$ {\sc ComputeSkylineBounded}$(P,s)$\\
			$s\qlet s^2$
		\qelihw\\
		\qreturn $\output$
	\end{algorithm}}}
	\caption{Algorithm to compute $\sky(P)$ in $O(n\log h)$ time.}
	\label{alg:OptimalComputeSkyline}
\end{figure}

\begin{theorem}
\label{thm:optimal}
	The skyline of a set of $n$ points can be computed in $O(n \log h)$ time,
	where $h$ is the number of points in the skyline.
\end{theorem}
\begin{proof}
	Consider the algorithm {\sc OptimalComputeSkyline} 
	in Figure~\ref{alg:OptimalComputeSkyline}.
	In the calls to {\sc ComputeSkylineBounded}$(P,s)$ we 
	have $s$ of the form $2^{2^r}$ for $r=1,2,3\dots$,
	until $h\le 2^{2^r}$. Thus, we finish when
	$r = \lceil \log_2(\log_2 h)\rceil$ and $s^{1/2} < h \le s$.
	Thus the algorithm uses
	\[
		\sum_{r=1}^{\lceil \log_2(\log_2 h)\rceil} O(n \log 2^{2^r})
		= \sum_{r=1}^{\lceil \log_2(\log_2 h)\rceil} O(n 2^r)
		= O\bigl( n 2^{\lceil \log_2(\log_2 h)\rceil} \bigr)
		= O(n \log h)
	\]
	time.
\end{proof}

A bit of thought shows that the exponent $2$ in Line 5
of {\sc OptimalComputeSkyline} can be replaced by any constant larger than $1$.
Thus, we could replace it by $s\qlet s^3$, for example.

\section{Optimization via computation of the skyline}
\label{sec:optimization_skyline}
In this section we show how to solve the problem $\opt(P,k)$ in $O(n \log h)$ time,
independently of the value of $k$.
First we show how to solve the decision problem in linear time, assuming
that $\sky(P)$ is available, and then consider the optimization problem.

\subsection{Decision problem}
Assume that $\sky(P)$ is already available and sorted by increasing $x$-coordinate.

For any point $p\in \sky(P)$ and any $\lambda\ge 0$,
the \DEF{next relevant point} for $p$ with respect to $\lambda$,
denoted by $\nrp(p,\lambda)$, is
is the point $q$ of $\sky(P)$ that is furthest from $p$
subject to the constraints that $d(p,q)\le \lambda$ and $x(q)\ge x(p)$.
See Figure~\ref{fig:nrp}.
Note that the next relevant point may be $p$ itself, as it may be the only
point of $\sky(P)$ that satisfies the conditions.
Because of the monotonicity property of Lemma~\ref{le:monotone},
we can find $\nrp(p,\lambda)$ scanning $\sky(P)$ from $p$ 
onwards; the running time is proportional to the number of points
of $\sky(P)$ between $p$ and $\nrp(p,\lambda)$.
(One could also use an exponential search, which is significantly
better when $k\ll n$.)

\begin{figure}
\centering
	\includegraphics[page=4,scale=1.1]{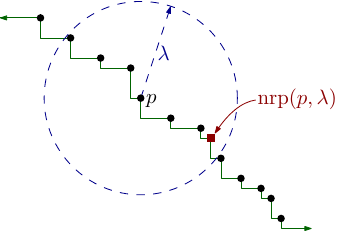}
	\caption{The next relevant point for $p$ with respect to $\lambda$ is indicated
		with a squared mark.}
	\label{fig:nrp}
\end{figure}

For a given $\lambda\ge 0$ we can decide in $O(h)$ time
whether $\opt(P,k)\le \lambda$.
The idea is a simple greedy approach that is used for 
several other ``1-dimensional'' problems: 
we iteratively compute next relevant points
to compute the centers.
Detailed pseudocode is in Figure~\ref{fig:DecisionSkyline1}.

\begin{lemma}
\label{le:decisionlinear}
	Given a skyline $S$ sorted by $x$-coordinate, an integer $k$ and a value $\lambda$, 
	we can decide in $O(h)$ time whether $\opt(S,k)\le \lambda$.
\end{lemma}
\begin{proof}
	Consider the algorithm {\sc DecisionSkyline1} in Figure~\ref{fig:DecisionSkyline1}.
	The input is an array $S[1\dots h]$ describing the skyline sorted by $x$-coordinate.
	We use an index $i$ that scans the skyline $S[\dots]$.
	At the $a$th iteration of the for-loop, we have an index $\ell_a$
	and we find the largest indices $c_a,r_a$ such that 
	$S[c_a]=\nrp(S[\ell_a],\lambda)$ and 
	$S[r_a]=\nrp(S[c_a],\lambda)$.
	(We use $\ell_a$, $c_a$ and $r_a$ because of \emph{left}, \emph{center} and 
	\emph{right} for the $a$th cluster.)
	
	With a simple inductive argument we have the following property:
	after $a$ iterations of the repeat loop,
	the index $i$ is the maximum index with the property
	that $S[1\dots (i-1)]$ can be covered with $a$ disks of radius $\lambda$
	centered at points of $S$.
	
	The procedure finishes when we use $k$ centers but the last cluster does
	not cover $S[h]$ because $r_k<h$,
	or when we use at most $k$ centers and we are covering $S[r_a]=S[h]$ because $i=h+1$.
	In the former case we have $\opt(S,k)>\lambda$,
	while in the latter case we have $\opt(S,k)\le \lambda$ and $\output$
	has the list of centers.
	
	The running time is $O(h)$ because the index $i$ only increases and between
	consecutive increases we spend $O(1)$ time.
\end{proof}

\begin{figure}
\centering
\fbox{~~~\parbox{0.95\linewidth}{
	\begin{algorithm}{{\sc DecisionSkyline1}}{
	\qinput{A skyline given as an array $S[1\dots h]$ sorted by increasing $x$-coordinate, a positive integer $k\le h$ and 
		a real value $\lambda\ge 0$.} 
	\qoutput{A solution $Q\subseteq S$ with $|Q|\le k$ and $\psi(Q,S)\le \lambda$, if $\opt(S,k)\le\lambda$,
	and ``\emph{incomplete}'' if $\lambda < \opt(S,k)$.}}
	$\output\qlet$ new empty list\\
	$i\qlet 1$ (* counter to scan $S[\cdot]$ *)\\
	\qfor $a=1,\dots, k$ \qdo\\
		$\ell_a\qlet i$ (* $S[\ell_a]$ is the first point to be covered by the $a$th center *)\\
		(* find $\nrp(S[\ell_a],\lambda)$ *)\\
		\qwhile $d(S[\ell_a],S[i])\le \lambda$ \qand $i\le h$ \qdo\\
			$i\qlet i+1$
		\qelihw\\
		$c_a\qlet i-1$ \qquad (* $S[c_a]$ is the center for the $a$th cluster *)\\
		(* find $\nrp(S[c_a],\lambda)$ *)\\
		\qwhile $d(S[c_a],S[i])\le \lambda$ \qand $i\le h$ \qdo\\
			$i\qlet i+1$
		\qelihw \\
		$r_a\qlet i-1$ \qquad (* dummy operation for easier analysis *)\\
		append $S[c_a]$ to the list $\output$ \quad (* $S[c_a]$ is a center that covers $S[\ell_a],\dots, S[r_a]$*)\\
		\qif $i>h$ \qthen (* we finished scanning $S$ because $S[i-1]=S[h]$ *)\\
			\qreturn $\output$
		\qfi
	\qrof\\
	\qreturn ``\emph{incomplete}'' 
	\end{algorithm}}}
\caption{Decision algorithm for testing whether $\opt(S,k)\le \lambda$. 
	The index $a$ is superfluous but helps analyzing the algorithm.}
\label{fig:DecisionSkyline1}
\end{figure}

\subsection{Optimization problem}
Let us turn our attention now to the optimization problem
of computing $\opt(P,k)$ and an optimal solution.

We start computing $\sky(P)$ explicitly and storing it in an array $S[1 \dots h]$
sorted by increasing $x$-coordinate.
This takes $O( n\log h)$ time using the optimal algorithm presented in 
Theorem~\ref{thm:optimal} or the algorithms in~\cite{KirkpatrickS85,Nielsen96}.

Consider the $h\times h$ matrix $M$ defined by
\[
	[M]_{i,j} =\begin{cases}
			d(S[i],S[j]) & \text{if $i<j$,}\\ 
			-d(S[i],S[j]) & \text{if $i\ge j$.}\\
		\end{cases}
\]
Lemma~\ref{le:monotone} implies that $M$ is a \DEF{sorted matrix}: each row has increasing values
and each column has decreasing values. The matrix $M$ is not constructed explicitly,
but we work with it implicitly.

Note that $\opt(P,k)$ is one of the (non-negative) entries of $M$
because $\opt(P,k)$ is an interpoint distance between points in $\sky(P)$.
We perform a binary search among the entries of $M$ to find the largest
entry $\lambda^*$ in $M$ such that $\opt(P,k)\le \lambda^*$.
We then have $\opt(P,k) = \lambda^*$.

To perform the binary search, we use an algorithm for the \emph{selection problem}
in sets of real numbers: given a set $A$ and a natural number $a$ with $1\le a\le |A|$,
we have to return the element with rank $a$ when $A$ is sorted non-decreasingly.
Frederickson and Johnson~\cite[Theorem 1]{FredericksonJ84} show
that the selection problem for the entries of a sorted matrix $M$ 
of size $h\times h$ can be solved in $O(h)$ time looking at $O(h)$ entries of $M$.
This means that, for any given $a$ with $1\le a \le h^2$, we can select
in $O(h)$ time the $a$th element of the entries of $M$.

This implies that we can make a binary search among the elements of $M$ 
without constructing $M$ explicitly.
More precisely, we make $O(\log (h^2))=O(\log h)$ iterations of the binary search,
where at each iteration we select the $a$th element $\lambda_a$ in $M$, 
for a chosen $a$,
and solve the decision problem in $O(h)$ time using Lemma~\ref{le:decisionlinear}.
In total we spend $O(h \log h)$ time. It is easy to obtain
an optimal solution once we have $\lambda^*=\opt(P,k)=\opt(S,k)$ by calling
{\sc DecisionSkyline1}$(S,\lambda^*)$.
We summarize.

\begin{theorem}
\label{thm:optimization1}
	Given a set $P$ of $n$ points in the plane and an integer parameter $k\le h$,
	we can compute in $O(n \log h)$ time $\opt(P,k)$ and an optimal solution,
	where $h$ is the number of points in the skyline $\sky(P)$.
	If $\sky(P)$ is already available, then $\opt(P,k)$ and an optimal solution
	can be found in $O(h \log h)$ time.
\end{theorem}

An alternative approach is explained and justified in detail
in Frederickson and Zhou~\cite[Lemma~2.1]{abs-1711-00599}.
The bottom line is that we can perform a binary search in $M$ using 
$O(h)$ time and looking at $O(h)$ entries in $M$, plus
the time needed to solve $O(\log h)$ decision problems.
We provide a high-level description and refer to~\cite{abs-1711-00599}
for the details, including pseudocode.
The algorithm is iterative and maintains a family $\mathbb{M}$
of submatrices of $M$. We start with $\mathbb{M}=\{ M \}$.
At the start of each iteration we have a family $\mathbb{M}$
of submatrices of $M$ that may
contain the value $\lambda^*$, are pairwise disjoint, 
and have all the same size.
In an iteration, we divide each submatrix in $\mathbb{M}$
into four submatrices of the same size, and perform two pruning steps, 
each using a median computation and solving a decision problem. 
At each iteration the size of the submatrices in $\mathbb{M}$ halves 
in each dimension
and the number of submatrices in $\mathbb{M}$ at most doubles.
After $O(\log h)$ iterations, the family $\mathbb{M}$
contains $O(h)$ matrices of size $1\times 1$, and
we can perform a traditional binary search among those values.
Actually, the algorithm of Frederickson and Johnson~\cite{FredericksonJ84}
for the selection problem follows very much the same paradigm, where
a counter is used instead of a decision problem.

Both papers~\cite{FredericksonJ84,abs-1711-00599}
include pseudocode. For a practical implementation it is probably
better to replace the deterministic linear-time computation of medians
by a simple randomized linear-time computation~\cite[Section 2.4]{DasguptaPV}.

\section{Optimization without computation of the skyline}
\label{sec:optimization_no_skyline}

We now describe another algorithm for the decision problem
and for the optimization problem where we do not have
$\sky(P)$ available, and we do not compute. 

\subsection{Decision problem}
Assume that we are given $P$, a positive integer $k$
and a real value $\lambda\ge 0$.
We want to decide whether $\opt(P,k)\le \lambda$.

We use the same idea that was exploited in the algorithm {\sc ComputeSkylineBounded}.
We split the point set into groups and compute the skyline of each group.
Then we note that we can use a binary search along the skyline of each group
to find the next relevant point on the skyline of each group.
First we provide a simple technical lemma and
then show how to compute the next relevant point along $\sky(P)$.
Details are provided as algorithm {\sc NextRelevantPoint} 
in Figure~\ref{fig:NextRelevantPoint}.

In our discussion we will be using the infinite curve $\alpha(p,\lambda)$  
obtained by concatenating the infinite upward vertical ray from $p+(\lambda,0)$,
the circular arc centered at $p$ and radius $\lambda$ from $p+(\lambda,0)$ to $p+(0,-\lambda)$
clockwise, and the vertical downward infinite ray from $p+(0,-\lambda)$.
See Figure~\ref{fig:alpha}. 
\begin{figure}
\centering
	\includegraphics[page=5,scale=1]{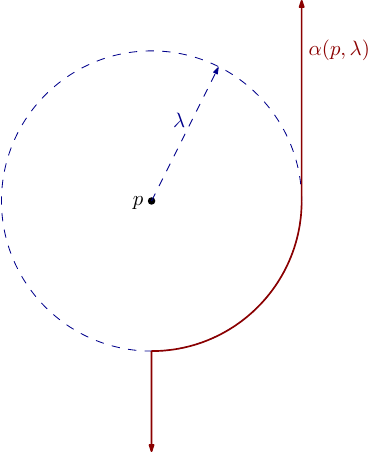}
	\caption{The curve $\alpha(p,\lambda)$.}
	\label{fig:alpha}
\end{figure}

\begin{lemma}
\label{le:binary}
	Given a skyline $Q$ sorted by $x$-coordinate such that $Q$ has points
	on both sides of $\alpha(p,\lambda)$, we can compute in $O(\log |Q|)$ time
	the point of $Q$ to the left of $\alpha(p,\lambda)$ and maximum $x$-coordinate
	and the point of $Q$ to the right of $\alpha(p,\lambda)$ and minimum $x$-coordinate. 
\end{lemma}
\begin{proof}
	Note that the set of points from $Q$ to the left of $\alpha(p,\lambda)$
	forms a contiguous subsequence of $Q$ (sorted by $x$-coordinate). 
	Moreover, for any point $q$, we can decide in constant time whether it
	is to the left or the right of $\alpha(p,\lambda)$.
	Thus, we can perform a binary search along $Q$ to find where the change occurs.
\end{proof}

\begin{lemma}
\label{le:binarymany}
	Assume that $P$ is split into $P_1,\dots P_t$,
	and for each $i\in \{ 1,\dots, t\}$ the points of $\sky(P_i)$ are sorted 
	by $x$-coordinate for binary searches.
	Let $p$ be a point of $\sky(P)$.
	Then we can find the next relevant point $\nrp(p,\lambda)$
	in $O(t+ \sum_i \log |P_i|)$ time. 
\end{lemma}	
\begin{proof}
	For each $i\in \{ 1,\dots, t\}$,
	let $q_i$ and $q'_i$ be the points along $\sky(P_i)$ immediately 
	before and after $\alpha(p,\lambda)$, respectively.
	The points $q_i,q'_i$ can be computed with a binary search in $O(\log |P_i|)$ time;
	see Lemma~\ref{le:binary}.

	Let $q$ and $q'$ be the points of $\sky(P)$ immediately 
	before and after $\alpha(p,\lambda)$, respectively.
	The task is to find $q=\nrp(p,\lambda)$.
	Let $\gamma$ be the $L$-shaped curve connecting 
	$q$ to $(x(q),y(q'))$ and then to $q'$.
	See Figure~\ref{fig:binarymany}.
	Let $P_i$ be the group that contains $q$ and let $P_j$ be the group that contains $q'$.
	If the horizontal segment connecting $(x(q),y(q'))$ to $q'$ crosses $\alpha(p,\lambda)$,
	then $q'$ is the highest point to the right of $\alpha(p,\lambda)$ (breaking ties
	in favor of larger $x$-coordinates).
	It follows that $q'=q'_j$ and moreover $q'$ is the highest point among $q'_1,\dots,q'_t$
	(breaking ties in favor of larger $x$-coordinates).
	If the vertical segment connecting $q$ to $(x(q),y(q'))$ crosses $\alpha(p,\lambda)$,
	then there is no point of $P$ to the right of $x=x(q)$ and to the left of $\alpha(p,\lambda)$.
	If follows that $q=q_i$ and moreover $q$ is the rightmost point among $q_1,\dots,q_t$ (breaking
	ties in favor of larger $y$-coordinates). 
	
	\begin{figure}
	\centering
		\includegraphics[page=6,scale=1.1]{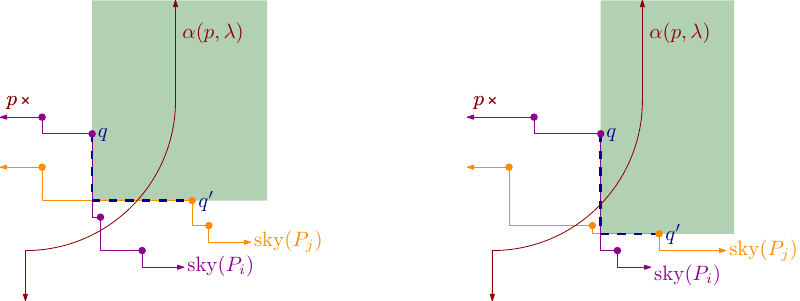}
		\caption{Cases in the proof of Lemma~\ref{le:binarymany}; 
			the shaded region is empty of points.
			Left: the case when the segment connecting $(x(q),y(q'))$ to $q'$ crosses $\alpha(p,\lambda)$.
			Right: the case when the segment connecting $q$ to $(x(q),y(q'))$ crosses $\alpha(p,\lambda)$.}
		\label{fig:binarymany}
	\end{figure}

	We have seen that $q$ is the rightmost point among $q_1,\dots,q_t$ 
	(breaking ties in favor of larger $y(\cdot)$) or
	$q'$ is the highest point among $q'_1,\dots,q'_t$, (breaking ties in favor of larger $x(\cdot)$),
	or both.
	As it can be seen in the examples of Figure~\ref{fig:binarymany}, 
	it may be that only one of the claims is true.
	
	Algorithmically we proceed as follows, using Lemmas~\ref{le:next}
	and~\ref{le:decisionmany} for some of the computations.
	We set $q_0$ to be the rightmost point among $q_1,\dots,q_t$ 
	(breaking ties in favor of larger $y(\cdot)$)
	and $q'_0$ to be highest point among $q'_1,\dots,q'_t$ 
	(breaking ties in favor of larger $x(\cdot)$).
	Using Lemma~\ref{le:decisionmany} we check whether $q'_0$ belongs to $\sky(P)$.
	If $q'_0$ belongs to $\sky(P)$, then we have $q'=q'_0$ and return 
	$\prev(\sky(P),x(q'_0))$, where we use Lemma~\ref{le:decisionmany}
	to compute $\prev(\,)$.
	If $q'_0$ does not belong to $\sky(P)$, then we have $q=q_0$ and return $q_0$.
	Details are provided in the algorithm {\sc NextRelevantPoint} 
	of Figure~\ref{fig:NextRelevantPoint}.
\end{proof}

\begin{figure}
\centering
\fbox{~~~\parbox{0.94\linewidth}{
	\begin{algorithm}{{\sc NextRelevantPoint}}{
	\qinput{$\sky(P_1),\dots,\sky(P_t)$ stored for binary searches, a point $p\in \sky(P_1\cup\dots \cup P_t)$
		and a real value $\lambda\ge 0$.} 
	\qoutput{$\nrp(p,\lambda)$ for $P=P_1\cup\dots \cup P_t$.}}
	\qfor $i=1,\dots,t$ \qdo\\
		$q_i\qlet$ last point in $\sky(P_i)$ to the left of or on $\alpha(p,\lambda)$, using binary search\\
		$q'_i\qlet \next(\sky(P_i),x(q_i))$ \qquad (* $q'_i$ to the right of $\alpha(p,\lambda)$ *)
	\qrof\\
	$q_0\qlet$ rightmost point among $q_1,\dots,q_t$, breaking ties in favor of larger $y(\cdot)$\\
	$q'_0\qlet$ highest point among $q'_1,\dots,q'_t$, breaking ties in favor of larger $x(\cdot)$\\
	$(X,r)\qlet$ {\sc TestSkylineAndPredecessor}$((\sky(P_1),\dots,\sky(P_t)),q'_0)$\\
	\qif $X=$ True\\
		\qthen \qreturn $r$\\
	\qelse\\ \qreturn $q_0$	
	\end{algorithm}}}
	\caption{Algorithm to find the next relevant point using  
		$\sky(P_1), \dots,\sky(P_t)$.}
\label{fig:NextRelevantPoint}
\end{figure}

Now we can solve the decision problem using at most $2k$ calls to the
function that finds the next relevant point.
As we did in previous cases, we add dummy points to avoid 
corner cases where the binary searches would have to return
null pointers.
Details are provided as algorithm {\sc DecisionSkyline2}
in Figure~\ref{fig:DecisionSkyline2}

For later use, we distinguish a preprocessing
part, which is independent of $\lambda$ and $k$,
and a decision part that depends on $\lambda$ and $k$. Indeed, later on we will be using
the preprocessing once and the decision part multiple times.

\begin{figure}
\centering
\fbox{~~~\parbox{0.94\linewidth}{
	\begin{algorithm}{{\sc DecisionSkyline2}}{
	\qinput{A set $P$ of points, positive integers $k$ and $\kappa\le n$
		and a real value $\lambda\ge 0$.} 
	\qoutput{A solution $Q\subseteq \sky(P)$ with $|Q|\le k$ and $\psi(Q,P)$, if $\opt(S,k)\le \lambda$,
	and ``\emph{incomplete}'' if $\lambda < \opt(S,k)$.}}
	{\bf (* Preprocessing *)}\\
	$p_0\qlet$ highest point of $P$, breaking ties in favor of larger $x(\cdot)$\\
	$q_0\qlet$ rightmost point of $P$, breaking ties in favor of larger $y(\cdot)$\\
	$\lambda_{\max}\qlet 1+d(p_0,q_0)$\qquad (* upper bound for $\opt(P,k)$ *)\\
	$M \qlet 2\lambda_{\max} + \max \bigcup_{p\in P}\{|x(p)|,|y(p)| \}$\\
	$t \qlet \lceil n/\kappa \rceil$\\
	split $P$ into $t$ groups $P_1,\dots,P_t$, each of at most $\kappa$ points\\
	\qfor $i=1,\dots,t$ \qdo\\
		append the points $(-M,M)$ and $(M,-M)$ to $P_i$ \qquad (* dummy end points *)\\
		$S_i \qlet$ {\sc SlowComputeSkyline}($P_i$)\\
		store $S_i$ in an array by increasing $x$-coordinate
	\qrof\\
	{\bf (* Decision *)}\\
	\qif $\lambda\ge \lambda_{\max}$ \qthen\\
		\qreturn $p_0$ (* or any other point of $\sky(P)$ *)
	\qfi\\
	$\output\qlet$ new empty list\\
	$\ell_1 \qlet p_0$ (* first non-dummy point of $\sky(P)$ *)\\
	\qfor $a=1,\dots, k$ \qdo\\
		$c_a\qlet$ {\sc NextRelevantPoint}$((S_1,\dots,S_t), \ell_a)$\\
		$r_a\qlet$ {\sc NextRelevantPoint}$((S_1,\dots,S_t), c_a)$\\
		append $c_a$ to the list $\output$ \qquad (* $c_a$ is a center that covers the portion of $\sky(P)$ from
			$\ell_a$ to $r_a$ *)\\
		(* compute $\ell_{a+1}=\next(\sky(P),x(r_a))$ *)\\
		\qfor $i=1,\dots,t$ \qdo\\
			$p_i\qlet \next(S_i,x(r_a))$, using binary search
		\qrof\\
		$\ell_{a+1}\qlet$ highest point among $p_1,\dots, p_t$, breaking ties in favor of larger $x(\cdot)$\\	
		\qif $x(\ell_{a+1})=M$ \qthen (* $\ell_{a+1}$ is the dummy point and we have finished *)\\
			\qreturn $\output$
		\qfi
	\qrof\\
	\qreturn ``\emph{incomplete}''  (* because $\lambda < \opt(S,k)$ *)
	\end{algorithm}}}
\caption{Decision algorithm for testing whether $\opt(S,k)\le \lambda$. 
	The index $a$ is superfluous but helps analyzing the algorithm.}
\label{fig:DecisionSkyline2}
\end{figure}

\begin{lemma}
\label{le:decisionfaster}
	Given a set $P$ of $n$ points and an integer parameter $\kappa\le n$,
	we can preprocess $P$ in $O(n\log \kappa)$ time such that,
	for any given real value $\lambda\ge 0$ and any positive integer $k$,
	we can decide in $O(k (n/\kappa)\log \kappa)$ time whether $\opt(S,k)\le \lambda$.
\end{lemma}
\begin{proof}
	See {\sc DecisionSkyline2} in Figure~\ref{fig:DecisionSkyline2}.
	Lines 1--11 correspond to the preprocessing, independent of $\lambda$ and $k$,
	while lines 12--27 correspond to the decision problem.
	The running time for the preprocessing is 
	$O(n+\sum_i |P_i|\log |P_i|)= O(n+ (n/\kappa) \kappa\log \kappa)= O( n \log \kappa)$.
	In the decision part we perform at most $k$ iterations,
	where each iteration takes $O(t)=O(n/\kappa)$ time plus the time for $O(1)$ calls
	to {\sc NextRelevantPoint} (Lemma~\ref{le:binarymany}) with $t=O(n/\kappa)$.
	Thus, each query takes
	$O(k)\cdot \bigl( O(n/\kappa)+ O((n/\kappa) \log \kappa) \bigr)$ time.
	The claim about the running times follow.
	
	Regarding correctness, the same argument that was used in the proof of
	Lemma~\ref{le:decisionlinear} applies.
	At the end of the $a$th iteration of the for-loop of lines 17--26
	we have the following properties:
	\begin{itemize}
		\item the points $c_1,\dots,c_a, \ell_1,\dots,\ell_a, r_1,\dots r_a$ 
			belong to $\sky(P)$;
		\item the point $r_a$ is the rightmost point of $\sky(P)$ with the property
			that the portion of $\sky(P)$ in $x\le x(r_a)$ can be covered
			with $a$ disks of radius $\lambda$ centered at points of $\sky(P)$;
		\item the disks centered at $c_1,\dots,c_a$ cover
			the portion of $\sky(P)$ in $x\le x(r_a)$.
	\end{itemize}
	This claim holds by induction because {\sc NextRelevantPoint}$(p,\lambda)$
	computes $\nrp(p,\lambda)\in \sky(P)$ whenever $p\in \sky(P)$; 
	see Lemma~\ref{le:binarymany}.
\end{proof}

Setting $\kappa=k$ in Lemma~\ref{le:decisionfaster} we obtain one of our main results,
where we see that computing the skyline, which takes $O(n\log h)$ time,
is not needed to solve the decision problem. The result is relevant
when $\log k = o(\log h)$; for example, when $k=\log h$.

\begin{theorem}
\label{thm:decisionfast}
	Given a set $P$ of $n$ points in the plane, a positive integer parameter $k$ and
	a real value $\lambda\ge 0$, 
	we can decide in $O(n \log k)$ time whether $\opt(P,k)\le \lambda$.
\end{theorem}
\begin{proof}
	We use the Algorithm of Lemma~\ref{le:decisionfaster},
	which is actually {\sc DecisionSkyline2} in Figure~\ref{fig:DecisionSkyline2},
	with $\kappa=k$.
	The preprocessing takes $O(n \log k)$ and deciding (the unique) $\lambda$
	takes 
	$O(k (n/k)\log k)= O(n\log k)$ time.
\end{proof}

\subsection{Optimization problem}
Now we turn to the optimization problem.
We use the paradigm of \emph{parametric search}; 
see for example~\cite{abs-1711-00599,Megiddo79,OostrumV04}.
Our presentation can be understood without a previous knowledge of the technique.

For convenience we set $\lambda^*=\opt(P,k)$.

We simulate running the decision algorithm {\sc DecisionSkyline2},
analyzed in Lemma~\ref{le:decisionfaster}, 
for the \emph{unknown} value $\lambda^*$.
Whenever we run into a computation that depends on the value of $\lambda^*$,
we apply a binary search to find an interval $(\lambda_1,\lambda_2]$
with the following properties:
(i) $(\lambda_1,\lambda_2]$ contains $\lambda^*$, and
(ii) the next step of the algorithm would be the same for any $\lambda\in (\lambda_1,\lambda_2]$.
We can then perform the next step of the algorithm for $\lambda^*$
as we would perform it for $\lambda_2$.

For the binary search we will use the following simple algorithm
to select the element with a given rank in a collection of sorted arrays.
As discussed in the proof,
the lemma is not optimal but it suffices for our purpose and
its proof is simple enough to implement it.

\begin{lemma}
\label{le:selection}
	Assume that we have $t$ arrays $S_1,\dots,S_t$, each with numbers 
	sorted increasingly. Set $S$ to be the union of the elements in $S_1,\dots, S_t$
	and let $n$ be the number of elements in the union.
	For any given value $\lambda^*$,
	we can find $\lambda'=\min \{ x\in S \mid x\ge \lambda^* \}$ 
	in $O(t\log^2 n)$ time plus the time needed to compare $O(\log n)$ times
	some element of $S$ with $\lambda^*$.
\end{lemma}
\begin{proof}
	We use a recursive algorithm. For simplicity we describe the algorithm
	assuming that the elements of $S_1,\dots,S_t$ are pairwise different.
	Breaking ties systematically we can assume this; for example replacing
	each element $S_i[j]$ with the triple $(S_i[j],i,j)$ and making
	comparisons lexicographically.
	We keep an \emph{active} contiguous portion of each $S_i$;
	we store it using indices to indicate which subarray is active.
	Thus, the recursive calls only takes $2t$ indices as input.
	
	Let $m_i$ be the median of the active subarray of $S_i$;
	assign weight $w_i$ to $m_i$, where $w_i$ is the number of active elements in $S_i$.
	We can find $m_i$ and $w_i$ in $O(1)$ time using arithmetic of indices.
	We select the weighted median $M$ of $m_1,\dots,m_t$ in $O(t)$ time.
	We compare $M$ against $\lambda^*$ to decide whether $\lambda^* \le M$ or $M< \lambda^*$.
	If $\lambda^* \le M$,
	then we clip the active part of each subarray $S_i$ to the elements at most $M$.
	If $M< \lambda^*$,
	then we clip the active part of each subarray $S_i$ to the elements at least $M$.
	Then we call recursively to search for $\lambda^*$ in the active subarrays.

	Note that at some point the active subarray of an array may be empty.
	This does not affect the approach; we just skip that subarray (or use weight
	$0$ for that).
	Clipping one single subarray can be done in $O(\log n)$ time using a binary search.
	Each call to the function takes $O(t\log n)$ time, plus one comparison
	between an element of $S$ and $\lambda^*$, plus the time for recursive calls.
	Since at each call we reduce the size of the active arrays by at least $1/4$,
	we make $O(\log n)$ recursive calls.  The result follows.
	
	It should be noted that there is also a very simple randomized solution
	that probably works better in practice:
	at each iteration choose one entry uniformly at random 
	among all active subarrays of the arrays, 
	compare the chosen element to $\lambda^*$, and clip the active subarrays accordingly.
	One can do faster deterministically~\cite{FredericksonJ82,FredericksonJ84,Kaplan0ZZ19}, 
	but it does not matter in our application. 
\end{proof}

We next provide the parametric version of Lemma~\ref{le:binarymany};
the algorithm {\sc ParamNextRelevantPoint} in Figure~\ref{fig:ParamNextRelevantPoint}
is the parametric counterpart of algorithm {\sc NextRelevantPoint}.

\begin{figure}
\centering
\fbox{~~~\parbox{0.94\linewidth}{
	\begin{algorithm}{{\sc ParamNextRelevantPoint}}{
	\qinput{$\sky(P_1),\dots,\sky(P_t)$ given as arrays $S_1[1\dots h_1],\dots, S_t[1\dots h_t]$
		sorted by increasing $x$-coordinate,
		and a point $p\in \sky(P_1\cup\dots \cup P_t)$.} 
	\qoutput{$\nrp(p,\lambda^*)$ for $P=P_1\cup\dots \cup P_t$.}}
	\qfor $i=1,\dots,t$ \qdo\\
		$q_i\qlet S_i[h_i]$ \qquad (* last point of $\sky(P_i)$ *)\\
		$a_i \qlet$ number of elements $q\in \sky(P_i)$ with $x(q)< x(p)$\\
		(* $S_i[a_i+1,\dots, h_i]$ stores $\{ q\in \sky(P_i)\mid x(p)\le x(q) \}$ *).
	\qrof\\
	$q_0\qlet$ rightmost point among $q_1,\dots,q_t$, breaking ties in favor of larger $y(\cdot)$\\
	(* $q_0$ is rightmost point of $\sky(P)$ *)\\
	(* boundary cases for the binary search *)\\
	\qif {\sc DecisionSkyline2}$(P, k, \kappa, 0) \neq$``\emph{incomplete}'' \qthen\\
		\qreturn $p$
	\qfi\\
	\qif {\sc DecisionSkyline2}$(P, k, \kappa, d(p,q_0)) =$``\emph{incomplete}'' \qthen\\
		\qreturn ``\emph{impossible}'' (* this never happens in our application *)
	\qfi\\
	(* binary search in sorted lists $\bigcup_i \Lambda_i = \bigcup_i \{ d(p,q)\mid q\in S_i[a_i+1,\dots, h_i] \}$ *)\\
	use Lemma~\ref{le:selection} to find 
		$\lambda'= \min \{\lambda \in \bigcup_i \Lambda_i \mid \lambda^*\le \lambda\}$ \\
	\qreturn {\sc NextRelevantPoint}($(\sky(P_1),\dots,\sky(P_t)), p, \lambda')$
	\end{algorithm}}}
	\caption{Algorithm to find the next relevant point for the unknown $\lambda^*$ 
		using $\sky(P_1), \dots,\sky(P_t)$.}
\label{fig:ParamNextRelevantPoint}
\end{figure}

\begin{lemma}
\label{le:parametricdecision}
	Assume that $P$ is split into $P_1,\dots P_t$,
	and for each $i\in \{ 1,\dots, t\}$ the points of $\sky(P_i)$ are 
	sorted by $x$-coordinate for binary searches.
	For any given point $p\in \sky(P)$ 
	we can compute the next relevant point $\nrp(p,\lambda^*)$
	in $O(t\log^2 n)$ time plus the time
	needed to solve $O(\log n)$ decision problems, without knowing $\lambda^*$.
\end{lemma}	
\begin{proof}
	Consider the algorithm {\sc ParamNextRelevantPoint} given 
	in Figure~\ref{fig:ParamNextRelevantPoint}.
	The input $\sky(P_i)$ is stored in an array $S_i[1\dots h_i]$.
	In line $3$ we perform a binary search along $S_i[\dots]$
	to find the index $a_i$ such that $S_i[1+a_i,\dots,h_i]$
	contains the points in $x\ge x(p)$.
	
	Let $\Lambda= \{ d(p,q)\mid q\in \sky(P),~ x(p)\le x(q)\}$
	and for each $i\in \{ 1,\dots, t\}$ let 
	$\Lambda_i=\{ d(p,q)\mid q\in \sky(P_i),~ x(p)\le x(q)\}$.
	Obviously we have $\Lambda \subseteq \bigcup_i \Lambda_i$
	and $0\in \Lambda$.
	The sets $\Lambda$ and $\Lambda_i$ are not constructed explicitly,
	but are manipulated implicitly.

	The boundary cases, when $\lambda^*=0=\min \Lambda$ and 
	when $\lambda^*> \max \Lambda$, are treated in lines 5--11.
	In the remaining case we have $0<\lambda^* \le \max \Lambda$.
	
	Consider any index $i\in \{ 1,\dots, t\}$.
	Since $p\in \sky(P)$, the point $p$ also belongs to $\sky(P_i\cup \{p\})$. 
	From Lemma~\ref{le:monotone} we then obtain that the distances from $p$
	increase along $\sky(P_i)$, assuming that we are on the same side 
	of the vertical line $x=x(p)$.
	It follows that 
	\[
		\forall j,j' \text{ with } a_i<j<j'\le h_i:~~ d(p,S_i[j])< d(p,S_i[j']).
	\]
	In short, the values of $\Lambda_i$, which are of the form $d(p,q)$,
	are sorted increasingly when $q$ iterates over $S_i[1+a_i,\dots,h_i]$.
	
	Since $\Lambda\subseteq \bigcup_i \Lambda_i$, we can perform
	a binary search in the union of the	$t$ sorted lists $\Lambda_1,\dots,\Lambda_t$
	to find the smallest element $\lambda'$ of $\bigcup_i \Lambda_i$ 
	with $\lambda^*\le \lambda'$.
	The next relevant point from $p$ for $\lambda^*$ and for $\lambda'$ is the same.
	Thus, we can return the next relevant point for $\lambda'$.

	To analyze the running time, we note that lines 1--5 perform $O(t)$
	binary searches, each with a running time of $O(\log n)$, and some other operations
	using $O(t)$ time. Thus, this part takes $O(t \log n)$ time.
	In lines 8--11 we are making two calls to the decision problem plus $O(1)$ time.
	In lines 12-13 we use Lemma~\ref{le:selection},
	which takes $O(t\log^2 n)$ time plus the time needed to solve 
	$O(\log n)$ decision problems.
\end{proof}

\begin{figure}
\centering
\fbox{~~~\parbox{0.94\linewidth}{
	\begin{algorithm}{{\sc ParametricSearchAlgorithm}}{
	\qinput{A set $P$ of points and a positive integer $k\le n^{1/4}$.} 
	\qoutput{A solution $Q\subseteq \sky(P)$ with $|Q|\le k$ and $\psi(Q,P)=\opt(P,k)$.}}
	{\bf (* Preprocessing -- same as in {\sc DecisionSkyline2} *)}\\
	$p_0\qlet$ highest point of $P$, breaking ties in favor of larger $x(\cdot)$\\
	$q_0\qlet$ rightmost point of $P$, breaking ties in favor of larger $y(\cdot)$\\
	$\lambda_{\max}\qlet 1+d(p_0,q_0)$\qquad (* upper bound for $\opt(P,k)$ *)\\
	$M \qlet 2\lambda_{\max} + \max \bigcup_{p\in P}\{|x(p)|,|y(p)| \}$\\
	$\kappa\qlet k^3 \log^2 n$\\
	$t \qlet \lceil n/\kappa \rceil$\\
	split $P$ into $t$ groups $P_1,\dots,P_t$, each of at most $\kappa$ points\\
	\qfor $i=1,\dots,t$ \qdo\\
		append the points $(-M,M)$ and $(M,-M)$ to $P_i$ \qquad (* dummy end points *)\\
		$S_i \qlet$ {\sc SlowComputeSkyline}($P_i$)\\
		store $S_i$ in an array by increasing $x$-coordinate
	\qrof\\
	{\bf (* Parametric part *)}\\
	$\output\qlet$ new empty list\\
	$\ell_1 \qlet p_0$ (* first non-dummy point of $\sky(P)$ *)\\
	\qfor $a=1,\dots, k$ \qdo\\
		$c_a\qlet$ {\sc ParamNextRelevantPoint}$((S_1,\dots,S_t), \ell_a)$\\
		$r_a\qlet$ {\sc ParamNextRelevantPoint}$((S_1,\dots,S_t), c_a)$\\
		append $c_a$ to the list $\output$ \qquad (* $c_a$ is a center that covers portion of $\sky(P)$ from
			$\ell_a$ to $r_a$ *)\\
		(* compute $\ell_{a+1}=\next(\sky(P),x(r_a))$ *)\\
		\qfor $i=1,\dots,t$ \qdo\\
			$p_i\qlet \next(S_i,x(r_a))$, using binary search
		\qrof\\
		$\ell_{a+1}\qlet$ highest point among $p_1,\dots, p_t$, breaking ties in favor of larger $x(\cdot)$\\	
		\qif $x(\ell_{a+1})=M$ \qthen (* $\ell_{a+1}$ is the dummy point and we have finished *)\\
			\qreturn $\output$
		\qfi
	\qrof
	\end{algorithm}}}
\caption{Algorithm to find an optimal solution to $\opt(P,k)$. 
	The index $a$ is superfluous but helps analyzing the algorithm.}
\label{fig:ParametricSearchAlgorithm}
\end{figure}

\begin{theorem}
\label{thm:optimization2}
	Given a set $P$ of $n$ points in the plane and an integer parameter $k$,
	we can compute $\opt(P,k)$ and an optimal solution in $O(n \log k + n \loglog n)$ time.
\end{theorem}
\begin{proof}
	If $k\ge n^{1/4}$, then $\log k = \Theta(\log n)$ 
	and we can use Theorem~\ref{thm:optimization1} to obtain an 
	algorithm with running time $O(n\log h)=O(n \log n) = O(n\log k)$.

	For the case $k< n^{1/4}$, we
	consider the algorithm {\sc ParametricSearchAlgorithm} given 
	in Figure~\ref{fig:ParametricSearchAlgorithm}.
	In lines 1--12 we make exactly the same preprocessing as in 
	the algorithm {\sc DecisionSkyline2} with $\kappa= k^3 \log^2 n \le n$.
	As stated in Lemma~\ref{le:decisionfaster},
	the preprocessing takes $O(n \log \kappa)= O(n(\log k + \loglog n))$ time
	and each subsequent decision problem can be solved
	in time
	\[
		O(k(n/\kappa)\log \kappa) = O((n/k^2 \log^2 n) (\log k + \loglog n))=O(n/k\log n).
	\]
	
	The rest of the algorithm, lines 13--25, follows the paradigm of 
	the decision part of {\sc DecisionSkyline2}.
	We make $k$ iterations, where at each $a$th iteration we compute 
	points $\ell_a,c_a,r_a\in \sky(P)$ such that $c_a=\nrp(\ell_a,\lambda^*)$
	and $r_a=\nrp(c_a,\lambda^*)$.
	The point $\ell_{a+1}$ is then set to be the point after $r_a$ along $\sky(P)$.
	
	We are making $O(k)$ calls to the function {\sc ParamNextRelevantPoint}
	and, as shown in Lemma~\ref{le:parametricdecision}, each of them takes 
	$O(t\log^2 n)$ time plus the time needed to solve $O(\log n)$ decision problems.
	In total we spend in all calls to {\sc ParamNextRelevantPoint}
	\[
		O(k)\cdot \bigl( O((n/k^3 \log^2 n) \log^2 n) + O(\log n) \cdot O(n/k\log n)  \bigr)
		= O(n)
	\]
	time. In lines 21-23 we spend additional $O(k)\cdot O(t\log \kappa)=O(n)$ time.
	The claim about the running time follows.
	
	Correctness follows the same line of thought as for algorithm {\sc DecisionSkyline2} 
	(see Lemma~\ref{le:decisionfaster}), using that {\sc ParamNextRelevantPoint}
	correctly finds $\nrp(p,\lambda^*)$ for the unknown value $\lambda^*=\opt(P,k)$.
	
	Note that the algorithm finds an optimal solution, but, as written,
	does not return the value $\opt(P,k)$.
	Assuming, for simplicity, that the returned solution has $k$ points,
	we then have 
	\[
		\opt(P,k)= \psi(P, \{ c_1,\dots,c_k \}) = 
			\max \bigl( \bigcup_{a=1,\dots,k} \{ d(c_a,\ell_a),d(c_a,r_a)\} \bigr). 
	\]
	This quantity can be computed through the iterations of the for-loop.
\end{proof}

\section{Algorithms for very small $k$}
\label{sec:smallk}
In this section we show that
the problem $\opt(P,k=1)$ can be solved in linear time; note that for
this running time we cannot afford to construct the skyline explicitly.
We also provide a $(1+\eps)$-approximation for $\opt(P,k)$ that is
fast when $k$ is very small.
The main tool is the following result.

\begin{lemma}
\label{le:bisector}
	Let $P$ be a set of $n$ points in the plane and let $p_0$ and $q_0$ be
	two distinct points of $\sky(P)$ with $x(p_0)<x(q_0)$.
	Consider the portion of the skyline 
	$S' = \{ p\in \sky(P)\mid x(p_0)\le x(p)\le x(q_0) \}$; this portion
	or the whole skyline is not available.
	In $O(n)$ time we can compute points $r_*$ and $r'_*$ of $S'$ such that
	\[
		r_* = \arg\min_{p\in S'}  ~\max\{ d(p,p_0),d(p,q_0) \} 
		~~\text{ and }~~
		r'_* = \arg\max_{p\in S'}  ~\min\{ d(p,p_0),d(p,q_0) \}. 
	\]
\end{lemma}
\begin{proof}
	We can assume that all the points $p$ of $P$ have
	$x(p_0)\le x(p)\le x(q_0)$ because other points can be ignored.
	Let $\beta$ be the bisector of $p_0q_0$; note that $\beta$ has positive slope and is not vertical.
	Let $p'$ be the point of $\sky(P)$ to the left of or on $\beta$ with largest $x$-coordinate.
	Let $q'$ be the point of $\sky(P)$ to the right of $\beta$ with smallest $x$-coordinate.
	We show that $\{ r_*,r'_*\}\subseteq \{ p',q' \}$
	and that $p',q'$ can be computed in linear time.
	The result follows because we can just evaluate $\max\{ d(p,p_0),d(p,q_0) \} $
	and $\min\{ d(p,p_0),d(p,q_0) \}$ for $p=p'$ and $p=q'$ and select the best.
	
	First we show that $r_*=p'$ or $r_*=q'$.
	Consider any point $p\in S'$ with $x(p)<x(p')$; we then have $x(p_0)\le x(p)<x(p')$.
	Because of Lemma~\ref{le:monotone} and because $p,p'$ are to the left of the bisector $\beta$ we have
	\begin{align*}
		\max \{ d(p,p_0), d(p,q_0) \} = d(p,q_0) > d(p',q_0) =  
			\max \{ d(p',p_0), d(p',q_0) \} ,
	\end{align*}
	which means that $p$ cannot be the optimal point $r_*$.
	A symmetric argument shows that any point $p\in S'$ with $x(q')<x(p)\le x(q_0)$ cannot be optimal.
	We conclude that $r_*$ is $p'$ or $q'$.
	A similar argument shows that $r'_*$ is also $p'$ or $q'$. 
	Indeed, for any point $p\in S'$ with $x(p_0)\le x(p)<x(p')$ we have 
	\begin{align*}
		\min \{ d(p,p_0), d(p,q_0) \} = d(p,p_0) < d(p',p_0) =  
			\min \{ d(p',p_0), d(p',q_0) \},
	\end{align*}
	and the argument for $p\in S'$ with $x(q')< x(p)\le <x(q_0)$ is similar.
	
	It remains to explain how to compute $p'$ and $q'$ in linear time.
	Let $p_1$ be the point to the left of or on $\beta$ with largest $x$-coordinate
	(breaking ties in favor of larger $y$-coordinate),
	and let $q_1$ be the point to the right of $\beta$ with largest $y$-coordinate
	(breaking ties in favor of larger $x$-coordinate).
	
	The same argument that was used in the proof of Lemma~\ref{le:binarymany}
	shows that $p_1\in \sky(P)$ or $q_1\in \sky(P)$ (or both). 
	We repeat the argument to make the proof self-contained. See Figure~\ref{fig:smallk}
	Consider the last point $p'\in \sky(P)$ to the left of or on $\beta$ and the first
	point $q'\in \sky(P)$ to the right of $\beta$.
	Define $\gamma$ as the $L$-shape curve connecting $p'$ to $(x(p'),y(q'))$ and to $q'$.
	If $\gamma$ crosses $\beta$ at the vertical segment connecting $p'$ to $(x(p'),y(q'))$,
	then the point $p'$ has to be the rightmost point to the left of or on $\beta$ and thus $p'=p_1$.
	If $\gamma$ crosses $\beta$ at the horizontal segment connecting $(x(p'),y(q'))$ to $q'$,
	then $q'$ has to be the topmost point to the right of $\beta$ and thus $q'=q_1$.

	\begin{figure}
	\centering
		\includegraphics[page=1,scale=1.1]{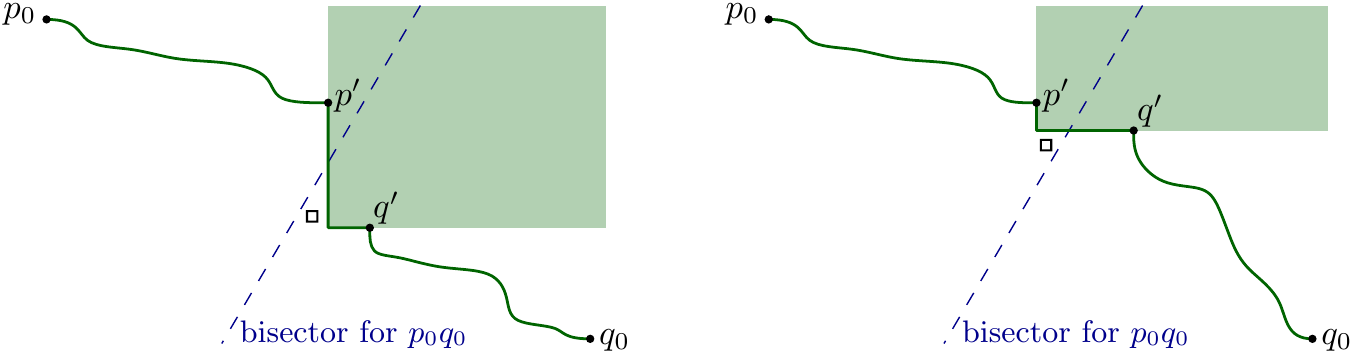}
		\caption{Cases in the proof of Lemma~\ref{le:bisector}; 
			the shaded region is empty of points.
			Left: the case when the segment connecting $p'$ to $(x(p'),y(q'))$ crosses $\beta$;
				the square represents a point showing that sometimes $q_1\notin \sky(P)$.
			Right: the case when the segment connecting $(x(p'),y(q'))$ to $q'$ crosses $\beta$;
				the square represents a point showing that sometimes $p_1\notin \sky(P)$.}
		\label{fig:smallk}
	\end{figure}

	We check whether $q_1$ belongs to $\sky(P)$ in linear time.
	Checking this amounts to checking whether $y(q_1)$ is the unique maximum
	among $y(q)$ for all the points $q\in P$ with $x(q_1)\le x(q)\le x(q_0)$.
	If $q_1\in \sky(P)$, then it must be $q_1=q'$ and 
	we can compute $p'$ using that
	\[
		x(p') = \max \{ x(p)\mid x(p)<x(q_1), ~y(p)>y(q_1) \}
	\]
	
	If $q_1 \notin \sky(P)$, we infer that $p_1\in \sky(P)$ and thus $p'=p_1$.
	We can compute the point $q'$ using that 
	\[
		y(q') = \max \{ y(p)\mid x(p_1) < x(p) \}.
	\]
	
	The whole procedure to find $p'$ and $q'$, as described, can be implemented to take linear time because 
	we only need $O(1)$ scans of the point set $P$.
\end{proof}

The $1$-center problem can be solved in linear time using Lemma~\ref{le:bisector}.

\begin{theorem}
\label{thm:k=1}
	Given a set $P$ of $n$ points in the plane,
	we can compute in $O(n)$ time $\opt(P,1)$ and an optimal solution.
\end{theorem}
\begin{proof}
	Let $p_0$ be the point with largest $y$-coordinate 
	(breaking ties in favor of larger $x$-coordinate)
	and let $q_0$ be the point with the largest $x$-coordinate
	(breaking ties in favor of larger $y$-coordinate). 
	The points $p_0$ and $q_0$ are the
	extreme points of $\sky(P)$. Note that for each point $p\in \sky(P)$
	we have $\psi(\{ p \},P) = \max \{ d(p,p_0), d(p,q_0) \}$
	because of Lemma~\ref{le:monotone}.

	If $p_0=q_0$, we return $p_0$ as the solution,
	which has cost $0$.
	If $p_0\neq q_0$, we apply Lemma~\ref{le:bisector} to compute 
	\[
		r_* = \arg\min_{p\in \sky(P)}  ~\max\{ d(p,p_0),d(p,q_0) \} = \arg\min_{p\in \sky(P)}  ~\psi(\{ p \},P)
	\]
	in linear time, and return $r_*$.
\end{proof}		
	
We next provide a $2$-approximation for $\opt(P,k)$
that is relevant when $k$ is very small.
As soon as $k=\Omega(\loglog n)$, Theorem~\ref{thm:optimization2} is better, as it can
compute an optimal solution.

\begin{lemma}
\label{le:k}
	Given a set $P$ of $n$ points in the plane and a positive integer $k$,
	we can compute in $O(kn)$ time a feasible solution $Q\subseteq \sky(P)$ 
	with at most $k$ points such that $\psi(Q,P)\le 2\cdot\opt(P,k)$.
	In the same time bound we also get $\psi(Q,P)$.
\end{lemma}
\begin{proof}
	We assume $k\ge 2$ because the case $k=1$ is covered by Theorem~\ref{thm:k=1}.
	
	Let $p_0$ be the point with largest $y$-coordinate 
	(breaking ties in favor of larger $x$-coordinate)
	and let $q_0$ be the point with the largest $x$-coordinate
	(breaking ties in favor of larger $y$-coordinate). 
	The points $p_0$ and $q_0$ are the extreme points of $\sky(P)$.
	
	Set $c_1=p_0$, $c_2=q_0$, and for $i=3,\dots k$,
	let $c_i$ be the point of $\sky(P)$ that is furthest from $c_1,\dots,c_{i-1}$.
	A classical result by Gonzalez~\cite{Gonzalez85} included in 
	textbooks (for example~\cite[Section 2.2]{WilliamsonS})
	shows that $C=\{ c_1,\dots,c_k\}$ is a $2$-approximation:
	we have $\psi(C,P)\le 2\cdot \opt(P,k)$.
	
	The points in $C$ can be computed iteratively in $O(kn)$ time as follows.
	Assume that we computed $c_1,\dots,c_i$ and we want to compute $c_{i+1}$.
	The vertical lines though $c_1,\dots c_i$ split the region $x(p_0)\le x \le x(q_0)$ into
	slabs. We maintain for each slab $\sigma$ the points of $P$ that are inside $\sigma$.
	For each slab $\sigma$, let $c(\sigma)\in \{c_1,\dots,c_i \}$ be the point
	that defines its left boundary and let $c'(\sigma)\in \{c_1,\dots,c_i \}$ be
	the point defining its right boundary. Obviously $x(c)< x(c')$. 
	We can use Lemma~\ref{le:bisector} to compute the point
	\[
		r'_*(\sigma) = \arg\max_{p\in \sky(P)\cap \sigma}  ~\min\{ d(p,c(\sigma)),d(p,c'(\sigma)) \}.
	\]
	Because of Lemma~\ref{le:monotone}, for each slab $\sigma$ and each $p\in \sky(P)\cap \sigma$
	we have
	\[
		\max \{ d(p,c_1), d(p,c_2),\dots, d(p,c_i) \} = \max \{ d(p,c(\sigma)), d(p,c'(\sigma)).
	\]
	If follows that the point $c_{i+1}$ has to be one of the points $r'_*(\sigma)$, where $\sigma$
	iterates over all slabs defined by $c_1,\dots,c_i$.
	Thus we take $c_{i+1}$ to be the point achieving the maximum 
	\[
		\max_{r'_*(\sigma)} ~ \min\{ d(r'_*(\sigma),c(\sigma)),d(r'_*(\sigma),c'(\sigma)) \}.
	\]
	
	Let $\sigma^*$ be the slab defining $c_{i+1}$; thus $c_{i+1}=r'_*(\sigma^*)$.
	Then the slab $\sigma^*$ has to be split into two subslabs with the vertical line $x=x(c_{i+1})$,
	and the points of $P\cap \sigma^*$ have to be rearranged into the two new subslabs.
	Since the slabs are interior disjoint, the 
	use of Lemma~\ref{le:bisector} over all slabs takes linear time.
	The rest of the work to compute $c_{i+1}$ and update the split of $P$ into
	slabs also takes linear time. We conclude that the construction
	of $C$ takes $O(kn)$ time.
	
	Note that with an extra round, we can compute within each slab
	the point that is furthest from any point of $C$. This would be the steps to compute
	$c_{k+1}$. Computing the distance from $c_{k+1}$ to $C$ we obtain $\psi(C,P)$.
\end{proof}

From a $2$-approximation we can compute a $(1+\eps)$-approximation 
using binary search.
This is relevant when $k=o(\loglog n)$, as otherwise
Theorem~\ref{thm:optimization2} computes an optimal solution.

\begin{theorem}
\label{thm:k}
	Given a set $P$ of $n$ points in the plane, a positive integer $k$ and a real
	value $\eps$ with $0<\eps<1$,
	we can compute in $O(kn+ n \loglog (1/\eps))$ time a feasible solution $Q$ with 
	at most $k$	points such that $\psi(Q,P)\le (1+\eps) \opt(P,k)$.
\end{theorem}
\begin{proof}
	We use Lemma~\ref{le:k} to compute a value $\lambda$ such
	that $\lambda \le 2 \opt (P,k) \le 2\lambda$. This takes $O(kn)$ time.
	Then we perform a binary search among the $O(1/\eps)$ values
	\[ 
		\lambda,\, \lambda+\tfrac{\eps\lambda}{2}, \, \lambda+2\tfrac{\eps\lambda}{2},\, 
		\lambda+3\tfrac{\eps\lambda}{2}, \, \lambda+4\tfrac{\eps\lambda}{2},
		\dots,\, \approx 2\lambda
	\]
	to find the index $j$ such that
	\[
		\lambda+(j-1)\tfrac{\eps\lambda}{2} ~<~ \opt(P,k) ~\le~ \lambda+j\tfrac{\eps\lambda}{2}.
	\]
	For this we have to solve $O(\log (1/\eps))$ decision problems.	
	Once we have the index $j$, we solve
	the decision problem for $\lambda+j\tfrac{\eps\lambda}{2}$ and return the feasible solution $Q$
	that we obtain.
	This is a $(1+\eps)$ approximation because
	\[
		\psi(Q,P) ~\le~ \lambda+j\tfrac{\eps\lambda}{2} ~=~
		\Bigl( \lambda+(j-1)\tfrac{\eps\lambda}{2} \Bigr) + \tfrac{\eps\lambda}{2} 
		<  \opt(P,k) + \eps  \opt(P,k).
	\]
	
	To analyze the running time of this step,
	we note that we have to solve $O(\log (1/\eps))$
	decision problems, all for the same number of points $k$.
	Using Theorem~\ref{le:decisionfaster} with 
	$\kappa = k^2 \log^2 (1/\eps)$, we spend 
	\[
		O\bigl( n \log (k^2 \log^2 \tfrac{1}{\eps} ) \bigr) + O(\log \tfrac{1}{\eps} ) \cdot O\Bigl(k \frac{n}{k^2\log^2 \tfrac{1}{\eps}} 
			\log (k^2 \log^2 \tfrac{1}{\eps})   \Bigr) 
		= O(n (\log k + \loglog \tfrac{1}{\eps}))
	\]
	time to solve the $O(\log (1/\eps))$ decision problems.
	The term $O(n\log k)$ is absorbed by $O(kn)$.
\end{proof}

The theorem implies that for any constant $k$ we can compute a $(1+\eps)$-approximation
in $O(n \loglog (1/\eps))$ time.

\section{Discussion}
\label{sec:conclusions}
We have improved previous results for computing the distance-based representatives
of the skyline in the plane, a problem that is relevant in the context of databases,
evolutionary algorithms and optimization.
We have shown that such representatives can be computed
without constructing the skyline, which is a conceptual shift with respect to previous works and approaches.
For example, with a relatively simple algorithm
we can solve the decision problem in $O(n\log k)$ time, which is asymptotically optimal.
We also provided algorithms for the optimization problem that are asymptotically optimal 
for a large range of values of $k$:
Theorem~\ref{thm:optimization1} is optimal when $\log k=\Omega ( \log h)$ 
and Theorem~\ref{thm:optimization2} is optimal when $k= \Omega(\log n)$.

Most of our algorithms are easy enough to be implemented.
Specially simple is the new decision algorithm that does not require computing the skyline.
In several cases, the decision combined with a trivial binary search may suffice
to get reasonable results. A practical implementation would use
randomized algorithms for selection, instead of deterministic linear-time median finding.

One may wonder whether using the Euclidean distance is a reasonable choice. That
depends on the application. The approach can be modified easily to work with other 
distances, such as the $L_\infty$ or the $L_1$-metric.
The main property that we need is that any disk (in the chosen metric) centered
at a point of the skyline intersects the skyline in a connected subpiece.
From this we can infer monotonicity of the distances (Lemma~\ref{le:monotone})
and we can perform binary searches along the skyline.

We next describe two research directions where progress is awaiting.
The first question is whether we can compute $\opt(P,k)$ in $O(n\log k)$ for all values $k$.
It may be good to start considering the case when $k$ is constant.
For example, can we compute $\opt(P,15)$ in $O(n)$ time when $h=\Theta(\sqrt{n})$?

One can imagine scenarios where we want to solve $\opt(P,k)$
for several different values $k$. Can we exploit correlation
between the problems in a non-trivial way?
More precisely, what is the computational complexity of the following problem:
given a set $P$ of points in the plane and a set $K\subseteq \{ 1,\dots, n\}$,
compute $\opt(P,k)$ for all $k\in K$.

\bibliographystyle{plainurl}
\bibliography{biblio}
\end{document}